\documentclass[11pt]{article}
\usepackage{natbib,amssymb,amsmath,fullpage,hyperref,graphicx}
\usepackage{amsthm}
\usepackage{gensymb}
\usepackage{caption,textcomp}
\usepackage{subcaption}

\newcommand{\Exp}[1]{{\text{E}}[ \ensuremath{ #1 } ]  }
\newcommand{\Var}[1]{{\text{Var}}[ \ensuremath{ #1 } ]  }

\newcommand{\bs}[1]{{\boldsymbol #1}}

\newtheorem{theorem}{Theorem}[section]

\newtheorem{lemma}[theorem]{Lemma}

\newtheorem{proposition}{Proposition}[section]
\newtheorem*{lemma*}{Lemma}

\begin{document}

\title{Adaptive multigroup confidence intervals with constant coverage}
\author{Chaoyu Yu$^{1}$ and  Peter D. Hoff$^{2}$ \\
 $^1$Department of Biostatistics, University of Washington-Seattle \\
 $^2$Department of Statistical Science, Duke University}
\maketitle


\begin{abstract}
Confidence intervals for the means of  multiple normal populations
are often based on a  hierarchical
normal model.
While commonly used interval procedures 
based on such a model
have the nominal coverage rate
on average across a population of groups,
their actual coverage rate for a given group 
will be above or below the nominal rate,  
depending on 
the value of the group mean. 
Alternatively, a coverage rate  that is constant as a function of 
a group's mean 
can be simply achieved
by using
a standard $t$-interval, based on data only from that group. 
The standard $t$-interval, however, fails
to share information across the groups and is therefore
not adaptive to easily obtained information about the
distribution of group-specific means.

In this article we construct
confidence intervals that  
have a constant frequentist coverage rate 
and that 
make use of information about
across-group heterogeneity, resulting in 
constant-coverage 
intervals that 
are narrower 
than standard $t$-intervals 
on average across groups. 
Such intervals are constructed by inverting biased
tests for the mean of a normal population.
Given a prior distribution on the mean,
Bayes-optimal biased tests can be inverted
to form Bayes-optimal confidence intervals with
frequentist coverage that is constant as a function of the mean. 
In the context of multiple groups, the prior distribution is replaced 
by a model of across-group heterogeneity. The parameters 
for this model can 
be estimated using data from all of the groups, and 
used to obtain
confidence intervals 
with constant group-specific coverage 
that adapt to information about the
distribution of group means.

\smallskip

\end{abstract}


\section{Introduction}
A commonly used experimental design is the one-way layout,
in which a random sample $Y_{1,j},\ldots, Y_{n_j,j}$
is obtained from
each of
several related groups $j\in \{1,\ldots, p\}$.
The standard normal-theory model for data from such a design
is that $Y_{1,j},\ldots, Y_{n_j,j}\sim$ i.i.d.\ $N(\theta_j,\sigma^2)$, 
independently across groups. 
Inference for the $\theta_j$'s
typically proceeds in one of two ways.
The ``classical'' approach is to use the
unbiased sample mean $\bar y_j$ as an estimator of $\theta_j$,
and to construct a confidence interval for $\theta_j$
by inverting the appropriate uniformly most powerful
unbiased (UMPU) test, that is, constructing the standard $t$-interval.
Such an approach essentially makes inference
for each $\theta_j$ using only data from group $j$
(although a pooled-sample estimate of $\sigma^2$ is often used).
The estimator of each $\theta_j$ is unbiased, and 
the confidence interval for each $\theta_j$ has the 
desired coverage rate.

An alternative approach is to utilize data from all of the groups
to infer each individual $\theta_j$. This is typically done by invoking
a hierarchical model, that is, 
a statistical model that describes the heterogeneity across the groups.
The standard one-way random effects model posits that 
$\theta_1,\ldots, \theta_p$ are a random sample from a normal
population, so that
$\theta_1,\ldots, \theta_p \sim $ i.i.d.\ $N(\mu,\tau^2)$.
In this case, shrinkage estimators
of the form
\[
\hat \theta_j  = \frac{ \hat \mu/\hat \tau^2 + \bar y_j n_j/\hat \sigma^2}
                      { 1/\hat \tau^2 + n_j/\hat \sigma^2 } 
\]
are often used, 
where $(\hat \mu, \hat \tau^2, \hat \sigma^2)$ are estimated using data 
from all of the groups. 
This estimator has a lower variance than the sample
mean, but is generally biased. 
Confidence intervals based on these shrinkage estimators
are often derived from the hierarchical model:
Letting $\tilde \theta_j $ be defined as
\[
\tilde\theta_j  = \frac{  \mu/ \tau^2 + \bar y_j n_j/ \sigma^2}
                      { 1/\tau^2 + n_j/ \sigma^2 }, 
\]
then
$\Exp{(\tilde \theta_j - \theta_j)^2} = (1/\tau^2 + n_j/\sigma^2)^{-1}$,
where the expectation integrates over both the normal model
for the observed data and the normal model representing heterogeneity
across the groups. 
This quantity is also the conditional variance of $\theta_j$ 
given data from group $j$, 
which suggests an empirical Bayes posterior interval for
$\theta_j$ of the form
$ \hat \theta_j  \pm t_{1-\alpha/2 } /
  \sqrt{1/\hat \tau^2 +n_j/\hat \sigma^2 }$, 
where $t_\gamma$ denotes the $\gamma$-quantile of the appropriate 
$t$-distribution. 
Compared to the classical $t$-interval
$\bar y_j \pm t_{1-\alpha/2} \sqrt{\hat\sigma^2/n_j}$,
this interval is narrower by a factor
of $\sqrt{ \hat \tau^2/(\hat \tau^2 + \hat\sigma^2/n ) }$.
However, its coverage rate is not $1-\alpha$
for all groups.
While  the rate tends to be near the
nominal level
on average across all groups,
the rate for a specific group $j$ will
depend on the value of $\theta_j$.
Specifically, the coverage rate will be
too low for  $\theta_j$'s far from
the overall average $\theta$-value,
and too high for $\theta_j$'s that are close to this
average (see, for example, \citet[Section 4.8]{snijders_bosker_2012}).
Other types of empirical Bayes posterior intervals have been developed by  \citet{Morris_1983}, \citet{Laird_1987}, \citet{He_1992} and  \citet{hwang_2009}. 
Like the interval obtained from the hierarchical normal
model,
these intervals are narrower than the standard $t$-interval but
fail to have the target coverage rate for each group.

In the related problem of confidence region construction for a vector
of normal means, several authors have pursued
procedures that dominate those based on UMPU test inversion
\citep{berger_1980,casella_hwang_1986}.
In particular, \citet{Tseng_1997} obtain a modified empirical Bayes
confidence region that has exact frequentist coverage but is also uniformly
smaller than the usual procedure.
In this article we pursue similar results for the problem of
multigroup confidence interval construction.
Specifically, we develop a confidence interval procedure
that has the desired coverage rate for every group,
but also adapts to the heterogeneity across groups,
thereby achieving shorter confidence intervals
than the classical  approach on average across groups.
More precisely, 
our goal is to obtain 
a multigroup confidence interval procedure  
$\{ C^1(\bs Y),\ldots, C^p(\bs Y)\}$,  based on data $\bs Y$ from 
all of the groups,
that attains the target
frequentist coverage rate for each group and 
all values of $\bs \theta=(\theta_1,\ldots,\theta_p)$, 
so that
\begin{equation}  
\Pr( \theta_j \in C^j(\bs Y) | \bs \theta) = 1-\alpha \ \
\forall \, \bs\theta \in \mathbb R^p,  \ \forall \, j\in \{1,\ldots, p\}, 
\label{eqn:coverprop}
\end{equation}
and is also more efficient than the standard $t$-interval
on average across groups, so that
\begin{equation}
\Exp{ | C^j(\bs Y)|  }   <  2 t_{1-\alpha/2 }, 
\label{eqn:shrinkprop}
\end{equation}
where $|C|$ denotes the width of an interval $C$, and 
the expectation is with respect to an unknown distribution 
describing the  across-group heterogeneity of the $\theta_j$'s.
The interval procedures we propose satisfy
the constant coverage property (\ref{eqn:coverprop}) exactly.
Property (\ref{eqn:shrinkprop}) will hold approximately, 
depending on what the across-group distribution is and how well it
is estimated. 

The intuition behind our procedure is as follows:
While the standard $t$-interval for a single group is
uniformly most accurate
among unbiased interval procedures (UMAU), it is not uniformly
most accurate  among all procedures. We define 
classes of biased hypothesis tests for a normal mean, inversion of
which generates $1-\alpha$ frequentist
$t$-intervals that are more accurate
than the standard UMAU $t$-interval for some values of the parameter space, but
less accurate elsewhere. The class of tests  can be chosen to minimize
an expected width with respect to a prior distribution for the population mean,
yielding the confidence interval procedure (CIP)
that is Bayes-optimal
among all CIPs that have $1-\alpha$ frequentist coverage.   
We call the Bayes-optimal frequentist procedure a 
``frequentist assisted by Bayes'' (FAB) interval 
procedure. 
In a multigroup setting, the ``prior'' for the population mean 
is replaced by a model for across-group heterogeneity. 
The parameters in this model can be estimated using data 
from all of the groups, yielding an empirical FAB
confidence interval procedure
that maintains a coverage rate 
that is constant as a function of the group means. 

Several authors have studied constant coverage CIPs in the single-group 
case that differ from the UMAU procedure. 
Such procedures generally make use of some sort 
of prior knowledge about the population mean.  
In particular, our work builds upon that of 
\citet{pratt_1963}, 
who studied the Bayes-optimal $z$-interval for the 
case that $\sigma^2$ is known. 
Other related work includes
\citet{farchione_kabaila_2008}  and 
\citet{kabaila_dilshani_2014}, who developed procedures 
that make use of non-probabilistic prior knowledge that 
the mean is near a pre-specified
parameter value (e.g.\ zero). Their procedures 
have shorter expected widths near 
this special value, 
but revert to the UMAU procedures
when the data are far from this point. 
\citet{evans_hansen_stark_2005} obtained minimax 
CIPs for cases where  prior knowledge takes the form of 
bounds on the parameter values. 

The FAB $t$-interval we construct is a straightforward extension of 
the Bayes-optimal $z$-interval developed by  \citet{pratt_1963}. 
In the next section, we
review the 
FAB $z$-interval of Pratt 
and extend the idea to construct a FAB $t$-interval  for the 
case that $\sigma^2$ is unknown. 
In Section 3 we use the FAB $t$-interval procedure to obtain 
group-specific confidence intervals 
that have constant  
coverage rates for all groups and all values of $\bs\theta$, and 
are also asymptotically optimal as the number of groups increases. 
In Section 4 we illustrate the use of the FAB interval procedure
with an 
example dataset, and compare its performance to that of 
the UMAU and empirical Bayes procedures
often used for multigroup data. 
A discussion follows in Section 5. 
Proofs are given in an appendix.

\section{FAB confidence intervals}

Consider a model for a random variable $Y$ that is indexed by a single 
unknown scalar parameter $\theta\in \mathbb R$. 
A $1-\alpha$ confidence region procedure (CRP) for $\theta$ based on $Y$ 
is a set-valued function 
$C(y)$ such that $\Pr( \theta \in  C(Y) | \theta) =1-\alpha$ for 
all $\theta\in \mathbb R$. 
As is well-known, a CRP
can be constructed by inversion of 
a collection of hypothesis tests. 
For each $\theta\in \mathbb R$, let $A(\theta)$ be the acceptance region 
of an $\alpha$-level test of $H_\theta: Y \sim P_\theta$ 
versus $K_\theta: Y\sim P_{\theta'}, \theta'\neq \theta$. 
Then $C(y) = \{ \theta: y \in A(\theta)\}$ 
is a $1-\alpha$ CRP.
We take the risk $R(\theta,C)$ of a  $1-\alpha$ CRP   
to be its expected Lebesgue measure
\[ R(\theta,C) = \int  \int 1(y \in A(\theta')) \  d\theta'   \ P_\theta(dy).  \] 
For our model of current interest, 
$Y\sim N(\theta,\sigma^2)$ with $\sigma^2$ known, 
there does not exist a  
CRP that uniformly minimizes this risk over all values of $\theta$.  
However,  there exist optimal CRPs  within 
certain subclasses of procedures. For example, the standard 
$z$-interval, given by $C_z(y) =( y+\sigma z_{\alpha/2} , y+\sigma z_{1-\alpha/2})$, 
minimizes the risk among all unbiased CRPs derived by inversion of 
unbiased tests of $H_\theta$ versus $K_\theta$, and so is 
the uniformly most accurate unbiased (UMAU) CRP. 

That the interval is  unbiased means 
$\Pr( \theta' \in C_{z}(Y) |\theta) \leq 1-\alpha$ for all $\theta$ and
$\theta'$, and that it is UMAU means 
$R(\theta,C_{z}) = 2\sigma  z_{1-\alpha/2}\leq  R(\theta,\tilde C)$
for any other unbiased CRP $\tilde C$ and every $\theta$.
But
while $C_{z}$ is best among unbiased CRPs,
the lack of a UMP test of $H_\theta$ versus $K_\theta$  means
there will be CRPs corresponding
to collections of biased level-$\alpha$ tests that
have lower risks than
$C_{z}$
for \emph{some}
values of $\theta$.  This suggests that if we have
prior information that
$\theta$ is likely to be near some value $\mu$,
we may
be willing to incur 
larger risks for $\theta$-values far from $\mu$ 
in exchange for small risks near $\mu$. 
With this in mind,  we consider  
the Bayes risk
$R(\pi, C) =  \int R(\theta,C)\, \pi(d\theta)$, 
where $\pi$ is a prior distribution that describes how close 
$\theta$ is likely to be to $\mu$. 
This Bayes risk may be related to the
marginal (Bayes) probability of accepting $H_\theta$ as follows:
\begin{align*}
R(\pi,C) = \int R(\theta,C)\, \pi(\theta)d\theta &= 
              \int \int \int  1(y\in A(\theta')) \, d \theta' \, P_\theta(dy) 
    \,  \pi(d\theta)    \\
&=\int \int \int  1(y\in A(\theta')) P_\theta(dy) \pi(d\theta)  \,    
   d\theta' \\
&= \int \Pr( Y\in A(\theta') ) \, d\theta'. 
\end{align*}
The Bayes-optimal $1-\alpha$ CRP is obtained by 
choosing $A(\theta)$ to minimize $\Pr(y\in A(\theta))$ for 
each $\theta\in\mathbb R$, or equivalently, 
to maximize the probability that $H_\theta$ is rejected 
under the prior predictive (marginal) distribution $P_\pi$ for $Y$ 
that is induced by $\pi$. 
This means that the optimal $A(\theta)$ is the acceptance region 
of the most powerful test of 
the simple hypothesis $H_\theta: Y \sim P_\theta$  
versus the simple hypothesis $K_\pi: Y \sim P_\pi$. 
The confidence region obtained by inversion of this 
collection of acceptance regions is Bayes optimal 
among all CRPs having $1-\alpha$ frequentist coverage. 
We describe such a procedure as ``frequentist, assisted by Bayes'', 
or FAB. 

Using this logic, 
\citet{pratt_1963} 
obtained and studied the Bayes-optimal
optimal CRP for the model $Y\sim N(\theta,\sigma^2)$  with $\sigma^2$ known 
and
prior distribution $\theta\sim N(\mu,\tau^2)$. 
Under this 
distribution for $\theta$, the marginal
distribution for $Y$ is  $N(\mu,\tau^2+\sigma^2)$. 
The Bayes-optimal 
CRP is therefore given by inverting acceptance regions $A(\theta)$ 
of the most powerful tests of 
$H_\theta: Y\sim N(\theta,\sigma^2)$ versus $K_\pi: Y\sim N(\mu,\tau^2+\sigma^2)$
for each $\theta$. 
This optimal CRP is an interval, the endpoints of which may 
be obtained by solving two nonlinear equations. 
We refer to this CRP   as
Pratt's FAB $z$-interval.

The procedure used to obtain the FAB $z$-interval, and
the form used by Pratt, are not immediately extendable
to the more realistic situation in which 
$Y_1\,\ldots, Y_n \sim \text{i.i.d} \ 
 N(\theta, \sigma^2)$ where both $\theta$ and $\sigma^2$
are unknown. The primary reason is that in this case
the Bayes-optimal acceptance region depends on
the unknown value of $\sigma^2$, or to put it another
way, the null hypothesis $H_\theta$ is composite.
However, the situation is not too difficult to
remedy: Below we re-express Pratt's $z$-interval
in terms of a function that controls where 
the type I error is ``spent''. 
We then define a class of $t$-intervals based on such 
functions, from which we obtain the Bayes-optimal $t$-interval
for the  case that $\sigma^2$ is unknown.

\subsection{The Bayes-optimal $w$-function} 
For the model $\{Y\sim N(\theta,\sigma^2), \ \theta\in \mathbb R\}$ 
 we may limit consideration 
of CRPs to those obtained by inverting collections of two-sided tests:
\begin{lemma}
Suppose the distribution of $Y$ belongs to a one-parameter exponential family with parameter $\theta \in \mathbb{R}$. For any confidence region procedure $\tilde C$
there exists a procedure $C$, obtained by inverting a collection 
of two-sided tests, 
that has the same coverage as $\tilde C$ and a risk less than or equal 
to that of $\tilde C$.
\label{lemma:twoside}
\end{lemma}
For the normal model of interest, 
an interval $A(\theta)=(\theta-\sigma u ,\theta-\sigma l)$ will be
the acceptance region of a two-sided level-$\alpha$ test if and only if
$u$ and $l$ satisfy 
$\Phi(u) - \Phi(l) = 1-\alpha$, or equivalently, 
if $u=z_{1-\alpha w}$ and $l = z_{\alpha (1-w)}$ 
for some value of $w\in (0,1)$, 
where $\Phi$ is the standard normal CDF and 
$z_{\gamma} = \Phi^{-1}(\gamma)$. 
It is important to note that the value of  $w$,
and thus $l$ and $u$,  can vary with $\theta$
and still yield a $1-\alpha$ confidence region:
Let $w:\mathbb R  \rightarrow (0,1)$ and define
$A_w(\theta)  = (\theta - \sigma z_{1-\alpha w(\theta)}, \theta-\sigma z_{\alpha(1-w(\theta))} )$.
Then for each $\theta$, $A_w(\theta)$ is  the
acceptance region of
a level-$\alpha$ test of
$H_\theta$
versus $K_\theta$. Inversion of $A_w(\theta)$
yields a $1-\alpha$ CRP
given by
\begin{equation}
 C_w(y) = \{ \theta : y+\sigma z_{\alpha (1-w(\theta))} < 
             \theta < y+\sigma z_{1-\alpha w(\theta)}   \}. 
 \label{eq:C}
\end{equation} 
This confidence region can be seen as a generalization of 
the usual UMAU $z$-interval, 
given by 
$
C_{1/2}(y) = \{ \theta : y+\sigma z_{\alpha /2 } < 
                \theta < y+\sigma z_{1-\alpha/2}   \}$, 
 corresponding to a constant $w$-function of 
$w(\theta)=1/2$.
Given a prior distribution for $\theta$, the Bayes-optimal 
$w$-function corresponds to the Bayes-optimal CRP. 
For the prior distribution $\theta\sim N(\mu,\tau^2)$ considered by Pratt,   
the optimal 
$w$-function  depends on $\psi=(\mu,\tau^2,\sigma^2)$ and 
is given as follows:
\begin{proposition}
Let $Y\sim N(\theta,\sigma^2)$,  $\theta\sim N(\mu,\tau^2)$  and  let
$w:\mathbb R \rightarrow (0,1)$. Then
$R(\psi , C_{w_\psi} ) \leq  R(\psi ,C_w)$
where $w_{\psi}(\theta)$ is given by 
$ w_\psi(\theta) = g^{-1} ( 2\sigma(\theta-\mu)/\tau^2 ) $
with $g(w) =  \Phi^{-1}(\alpha w) - \Phi^{-1}(\alpha(1-w) )$.  
The function $w_\psi(\theta)$ is a continuous strictly 
increasing function of $\theta$. 
\label{prop:C}
\end{proposition}

As stated in \citet{pratt_1963} but not proven, 
$C_{w_{\psi}}(y)$ is actually an 
interval for each $y\in \mathbb R$, and 
so $C_{w_\psi}$ is a 
confidence interval procedure (CIP). In fact, 
a CRP $C_w$ will be a CIP as long as the $w$-function is 
continuous and nondecreasing:
\begin{lemma} Let $w: \mathbb R \rightarrow (0,1)$ be a continuous 
nondecreasing function. 
Then the set $C_w(y) = \{ \theta : 
y+\sigma z_{\alpha (1-w(\theta))} < 
             \theta < y+\sigma z_{1-\alpha w(\theta)}   \}  $
is an interval and 
can be written as $(\theta^L,\theta^U)$, where $\theta^L$ and $\theta^U$ 
are solutions to 
$\theta^L = y+ \sigma z_{\alpha(1-w(\theta^L))}$ and 
$\theta^U = y+ \sigma z_{1-\alpha w(\theta^U)}$. 
\label{lemma:incint}
\end{lemma}
A bit of algebra shows that
Pratt's FAB $z$-interval  can be expressed as
$C_{w_\psi}  = (\theta^L, \theta^U)$, where 
$\theta^L$ and $\theta^U$ solve 
\begin{align*}
\theta^U &=\frac{y+\sigma \Phi^{-1}(1-\alpha+\Phi(\tfrac{y-\theta^U}{\sigma}))}
                 {1+ 2\sigma^2/\tau^2 }  + \mu \frac{2\sigma^2/\tau^2}
                 {   1+ 2\sigma^2/\tau^2    }  \\  
\theta^L &=\frac{y+\sigma\Phi^{-1}(\alpha - \Phi(\tfrac{\theta^L-y}{\sigma}))}
            {1+ 2\sigma^2/\tau^2 }  + \mu \frac{2\sigma^2/\tau^2}
                 {   1+ 2\sigma^2/\tau^2    }  . 
\end{align*}
Solutions to these equations can be found with a zero-finding algorithm, 
and noting the fact that 
$\theta^L < y+ \sigma z_{\alpha}$ and $y+\sigma z_{1-\alpha} < \theta^U $. 

Some aspects of the FAB $z$-interval procedure are 
displayed graphically in 
Figure \ref{fig:zint}. The left panel gives the $w$-functions 
corresponding to the Bayes-optimal 95\% CIPs
for $\sigma^2=1$, $\mu=0$ and 
$\tau^2 \in \{1/4,1,4\}$.  
At varying rates depending on $\tau^2$, 
the $w$-functions approach zero or one
as $\theta$ moves 
towards $-\infty$ and $\infty$, 
respectively. 
The level-$\alpha$ tests corresponding to these 
$w$-functions are ``spending''  more of their type I error 
on $y$-values that are likely under the 
$N(\mu,\sigma^2+\tau^2)$
prior predictive distribution of $Y$. 
This makes 
the intervals narrower than the usual interval when $y$ is near $\mu$, 
and wider when $y$ is far from $\mu$, 
as shown in the middle panel of the figure. 
In particular, 
at $y=\mu$, the 95\% FAB $z$-interval with $\tau^2=1/4$ has a 
width of 3.29, which is about 84\% of that of the UMAU interval.
Average performance across 
$y$-values is given by risk, or expected confidence interval width, 
displayed in the top right plot. Expected 
widths of the FAB $z$-intervals are 
lower than those of the UMAU intervals for values of $\theta$ near $\mu$
(15\% lower for $\theta=\mu$ and $\tau^2=1/4$),
but can be much higher 
for $\theta$-values far away from $\mu$, particularly 
for small values of $\tau^2$. Relative to 
small values of $\tau^2$, the larger value of $\tau^2=4$ 
enjoys better performance than the UMAU interval over a wider 
range of $\theta$-values, 
but the improvement is not as
large near $\mu$. 
Additional calculations (available from the replication code 
for this article) show that the performance of the FAB 
interval near $\mu$ improves as $\alpha$ increases, as
compared to the UMAU interval. 
For example, with $\tau^2=1/4$ and $\alpha=0.50$, 
the width of the FAB interval at $y=\mu$ is about 25\% of 
that of 
the UMAU interval, and its risk at $\theta=\mu$ is 60\% that of 
the UMAU interval.

\begin{figure}
\centerline{\includegraphics[width=6.5in]{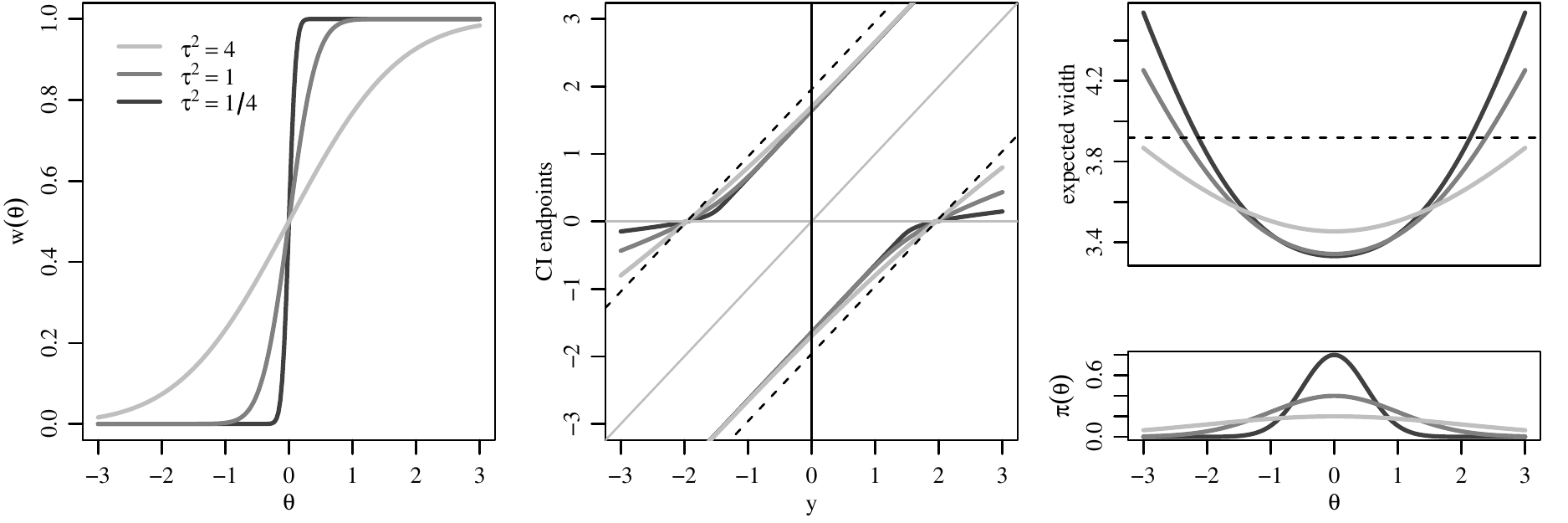}}  
\caption{Descriptions of the FAB $z$-procedure.  
The left plot gives  Bayes-optimal $w$-functions for three
values of $\tau^2$, at level $\alpha=0.05$.  
The middle plot gives the corresponding 
confidence interval procedures, with the UMAU procedure 
given by dashed lines.  
The top plot on the right gives the risk functions (expected widths) 
of the 95\% FAB $z$-intervals for the three values of $\tau^2$, 
with the corresponding prior densities plotted below.  } 
\label{fig:zint}
\end{figure}

\subsection{FAB $t$-intervals}
Adoption of Pratt's $z$-interval has been limited, possibly due to  
two factors: First, in most applications the population 
variance is unknown, and 
second, the prior distribution 
for $\theta$ must  be specified. 
We now address this first issue by developing 
a FAB $t$-interval. 
Suppose we have a sample $Y_1,\ldots, Y_n \sim 
 $ i.i.d.\ $N(\theta,\sigma^2)$, with sufficient 
statistics $(\bar Y, S^2)$, the sample mean and 
(unbiased) sample variance.
The standard UMAU $t$-interval is given 
by $\{ \theta:  \bar y + \tfrac{s}{\sqrt{n}} t_{\alpha/2}  < \theta < 
                \bar y + \tfrac{s}{\sqrt{n}} t_{1-\alpha/2}  \}$. 
This interval is symmetric around $\bar y$, with the same 
tail-area probability ($\alpha/2$) defining the lower and upper 
endpoints. 
The development of the 
$w$-function described in the previous subsection 
suggests viewing the UMAU $t$-interval as belonging to the larger class
of CRPs, given by 
\begin{align}  
C_w(\bar y,s^2) &=
\{ \theta:  \bar y + \tfrac{s}{\sqrt{n}} t_{\alpha(1-w(\theta))} 
               < \theta < 
                \bar y + \tfrac{s}{\sqrt{n}} t_{1-\alpha w(\theta)}  \},  
\label{eqn:citw}
\end{align}  
for some  $w:\mathbb R \rightarrow (0,1)$. Any procedure 
thus defined satisfies 
$\Pr( \theta \in C_w( \bar Y, S^2) | \theta )= 1-\alpha$ for any value of 
$\theta$. 
Additionally, $C_w$ is a CIP as long as $w$ is a 
continuous nondecreasing function:
\begin{lemma} Let $w: \mathbb R \rightarrow (0,1)$ be a continuous
nondecreasing function.
Then the set 
$C_w(\bar y,s^2) =
\{ \theta:  \bar y + \tfrac{s}{\sqrt{n}} t_{\alpha(1-w(\theta))} 
               < \theta < 
                \bar y + \tfrac{s}{\sqrt{n}} t_{1-\alpha w(\theta)}  \}$
is an interval and
can be written as $(\theta^L,\theta^U)$, where $\theta^L$ and $\theta^U$
are solutions to
$\theta^L = \bar y+ \tfrac{s}{\sqrt{n}}t_{\alpha(1-w(\theta^L))}$ and
$\theta^U = \bar y+ \tfrac{s}{\sqrt{n}}t_{1-\alpha w(\theta^U)}$.
\label{lemma:tincint}
\end{lemma}

For a given $w$-function, the endpoints of the interval  
can  be re\"expressed as 
\begin{align}
F \left (\frac{ \bar y - \theta^U }{s/\sqrt{n}}\right) & = \alpha w(\theta^U) \label{eqn:tupper} \\
F \left (\frac{\bar y-\theta^L }{s/\sqrt{n}}\right) &=  1- \alpha(1- w(\theta^L)),  \label{eqn:tlower}
\end{align} 
where $F$ is the CDF of the $t_{n-1}$ distribution. 
Using the same logic as at the beginning of Section 2, 
the Bayes risk of a CRP  for a prior distribution 
$\pi$ on $\theta$ and $\sigma^2$
is $R(\pi,C)  = \int \Pr( (\bar Y , S^2) \in A(\theta')) \, d\theta'$, 
where $\Pr( ( \bar Y , S^2) \in  A(\theta'))$ 
is the prior predictive (marginal)
probability of $(\bar Y, S^2)$ being in the acceptance 
region $A(\theta')$ under the prior distribution 
$\pi$. 
Given a prior $\pi$
that corresponds to a continuous, nondecreasing $w$-function, 
the Bayes-optimal FAB interval can be obtained
numerically by using an iterative algorithm to solve 
(\ref{eqn:tupper}) and (\ref{eqn:tlower}). However, 
this requires computation of the $w$-function, 
which for each $\theta$ is the minimizer in $w$  of 
$\Pr( (\bar Y , S^2) \in  A_w(\theta) )$, 
where 
\begin{equation}
A_w(\theta) =  
\left \{ (\bar y, s^2): t_{\alpha w} < \frac{\bar y - \theta }{s/\sqrt{n}} <  
    t_{1-\alpha(1-w) }     \right  \} . 
\end{equation}
Obtaining the optimal $w$-function will generally involve numerical 
integration. 
Consider a $N(\mu,\tau^2)$ prior on $\theta$ and  
so conditionally on
$\sigma^2$ we have $\bar Y \sim N( \mu , \sigma^2/n + \tau^2)$ 
and $(n-1)S^2/\sigma^2 \sim \chi^2_{n-1}$.  
From this we can show that $c (\bar Y -\theta)/(S/\sqrt{n}) $
has a noncentral $t_{n-1}$ distribution with noncentrality parameter
$\lambda= c \tfrac{\mu -\theta}{\sigma/\sqrt{n}}$, 
where $c =\sqrt{\sigma^2/n}/\sqrt{\sigma^2/n+\tau^2}$. 
Therefore, the  probability of the event 
$\{ (\bar Y , S^2) \in A(\theta)\}$, 
conditional on $\sigma^2$, can be written as
\[
\Pr( \{ \bar Y, S^2\}\in A(\theta)|\sigma^2) =
   F_\lambda( ct_{1-\alpha(1-w)}) - F_\lambda( ct_{\alpha w }), 
\]
where $F_\lambda$ is the CDF of the noncentral $t_{n-1}$ distribution 
with parameter $\lambda= c \tfrac{\mu -\theta}{\sigma/\sqrt{n}}$. 
The Bayes-optimal $w$-function is therefore given by 
\begin{equation}
 w_{\pi}(\theta) = \arg \min_w \int \left (  F_\lambda( ct_{1-\alpha(1-w)}) - F_\lambda( ct_{\alpha w }) \right ) 
      p_\pi(\sigma^2) \, d\sigma^2,  
\label{eqn:wtopt}
\end{equation}
where $p_\pi(\sigma^2)$ is the prior density over $\sigma^2$. 

In the replication 
material for this article we provide {\sf R}-code for obtaining $w_{\pi}(\theta)$ 
and the corresponding Bayes-optimal $t$-interval 
$C_{\pi}( \bar y,s^2)$ for 
the class of priors where $\theta$ and $\sigma^2$ are
\emph{a priori} independently distributed as normal and inverse-gamma  
random variables. 
Here, we provide some descriptions of this FAB $t$-interval procedure for 
some parameter values  that make the interval comparable to the 
$z$-interval from Section 2.1.
Specifically, we consider the case that $n=10$, 
$1/\sigma^2 \sim $ gamma$(1,10)$ and $\theta\sim N(0,\tau^2)$ for 
$\tau^2 \in \{1/4,1,4\}$.  
This makes the prior median of $\sigma^2$ near 10, and the variance of 
$\bar Y$ near 1 (and so the variance of $\bar Y$ here is comparable to 
the variance of $Y$ in Section 2.1). 
The left panel of Figure \ref{fig:tint} gives the $w$-functions, which 
are very similar to those of the FAB $z$-procedure displayed in 
Figure \ref{fig:zint}, but with 
somewhat larger derivatives near $\mu$. 
The second panel gives the FAB $t$-intervals as functions of 
$\bar y$, with $s^2$ fixed at 10. Again, the intervals 
resemble the corresponding $z$-intervals, but are slightly wider 
due to the use of $t$-quantiles instead of $z$-quantiles.  
The effect of not knowing $\sigma^2$ is more noticeable in the 
plot of the risk functions, given in the right-upper plot.
While the shapes of the risk functions are similar to 
those of the analogous $z$-intervals, the risks (expected widths) are
larger due to the fact that the width of a $t$-interval is 
dependent on $S^2$, which is proportional to a $\chi^2_9$ random 
variable having non-trivial skew.

\begin{figure}
\centerline{\includegraphics[width=6.5in]{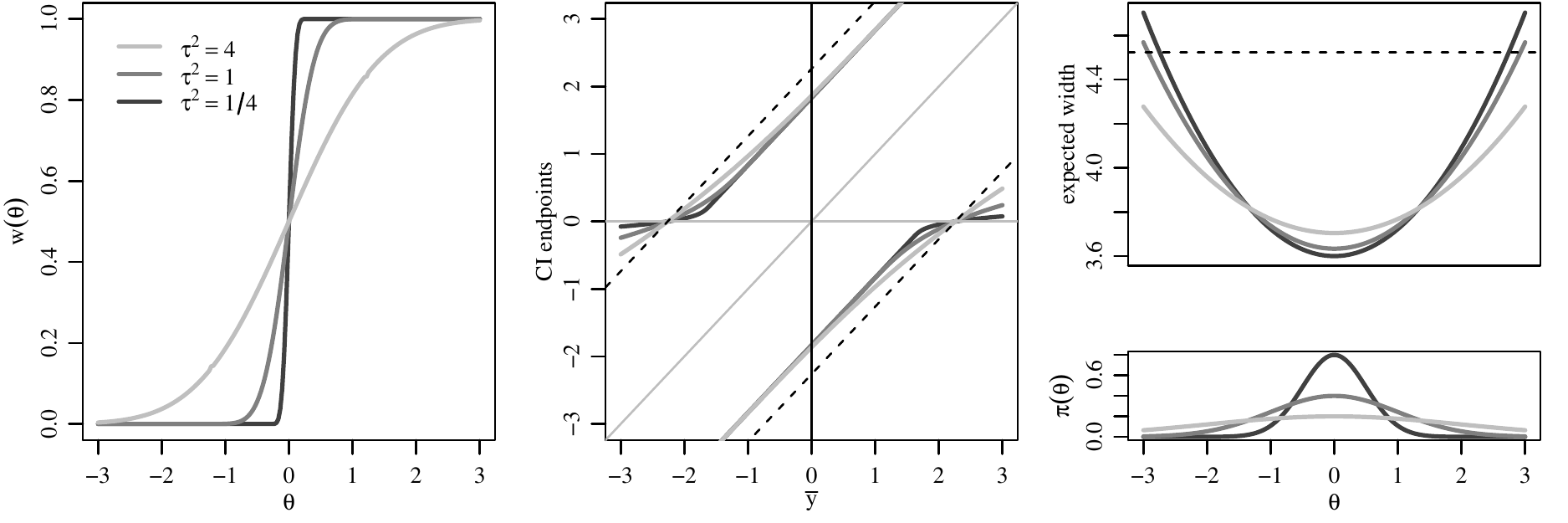}}
\caption{Descriptions of the FAB $t$-procedure.
The left plot gives Bayes-optimal $w$-functions for three
values of $\tau^2$, at level $\alpha=0.05$.
The middle plot gives the corresponding
confidence interval procedures with $s^2$ fixed at 10. 
The top plot on the right gives the expected widths
of the 95\% FAB $t$-intervals for the three values of $\tau^2$,
with the corresponding prior densities plotted below.  } 
\label{fig:tint}
\end{figure}

\section{Empirical FAB intervals for multigroup data}
A potential obstacle to the adoption of FAB confidence intervals 
is the aversion that many researchers have to specifying 
a distribution over $\theta$. 
However, in multigroup data settings,
probabilistic information about the mean $\theta_j$ of one group 
is may be obtained from data of the other groups. 
This information can be used to specify a probability 
distribution $\pi$ for the likely values of $\theta_j$, 
from which an empirical FAB interval may be constructed. 
Such an interval will have 
exact $1-\alpha$ coverage for every value of $\theta_j$, 
but  a shorter expected width for values that are 
deemed likely by $\pi$.  
For the usual homoscedastic hierarchical normal model having 
a common within-group variance, 
we develop such a procedure that may be 
used in practice, and show that it 
is risk-optimal 
asymptotically in the number of groups. 
We also develop a similar procedure for the case of 
heteroscedastic groups. 

\subsection{Asymptotically optimal procedure for homoscedastic groups}
Consider the case of $p$ normal populations 
with means $\theta_1,\ldots, \theta_p$ and common variance and sample size, 
so that $Y_{1,j},\ldots Y_{n,j} \sim $ i.i.d.\ $N(\theta_j,\sigma^2)$ 
independently 
across groups
(common sample sizes are used here solely to 
simplify notation). 
The standard hierarchical normal model posits that 
the heterogeneity across groups can be described by 
a normal distribution, so that 
$\theta_1,\ldots, \theta_p \sim  \text{i.i.d.} \, N(\mu,\tau^2)$. 
In the multigroup setting, 
this normal distribution  is not considered to be a prior distribution for
a single $\theta_j$, but instead is a statistical model 
for the across-group heterogeneity of $\theta_1,\ldots, \theta_p$. 
The parameters describing the across- and within-group
heterogeneity are $\psi=(\mu,\tau^2,\sigma^2)$.

For each group $j$ let $C^j$ be 
a $1-\alpha$ CRP for $\theta_j$ that possibly depends on 
data from all of the other groups. Letting $\bs C =\{ C^1,\ldots, C^p\}$  
we define 
the risk of such a multigroup confidence procedure as 
\[ 
  R(\bs C,\psi) = \frac{1}{p}\sum_{j=1}^p 
   \Exp{  |  C^j(\bs Y) | }, 
\]
where $\bs Y$ is the data from all of the groups and 
the expectation is over both  $\bs Y$ and 
$\theta_1,\ldots, \theta_p$. 
Under the hierarchical normal model, 
the risk at a value of $\psi$ 
is minimized 
by letting each $C^j$ be 
equal to $C_{w_\psi}(\bar y_j)$, 
the FAB $z$-interval 
defined in Section 2 but  with $\psi=( \mu,\tau^2,\sigma^2/n)$, 
since $\Var{\bar Y_j|\theta_j} = \sigma^2/n$. 
The oracle multigroup confidence procedure is 
then ${\bs C}_{w_\psi} = 
\{  C_{w_{\psi}}( \bar y_1 ),\ldots,   C_{w_\psi}( \bar y_p )  \},$
which has risk 
\[
 R( {\bs C}_{w_\psi},\psi) = \frac{1}{p}\sum_{j=1}^p 
   \Exp{  |   C_{w_\psi}(\bar y_j) | } = 
   \Exp{  |   C_{w_\psi}(\bar y) | } , 
\]
where 
$\bar y \sim N( \theta , \sigma^2/n)$
and $\theta\sim N(\mu,\tau^2)$. 
While this oracle procedure is generally unavailable in practice, 
estimates of $\psi$ may be obtained from the data and 
used to construct a multigroup CIP that achieves the 
oracle risk asymptotically as $p\rightarrow \infty$.  
To show how to do this, we first construct a $1-\alpha$ CIP 
for a single $\theta$  based on 
$\bar Y \sim N( \theta , \sigma^2/n)$ and 
 independent estimates $S^2$ and $\hat \psi$ of $\sigma^2$ 
and $\psi$. We show how the risk of this CIP converges to the 
oracle risk as 
$S^2\stackrel{a.s.}{\rightarrow} \sigma^2$ and 
$\hat \psi\stackrel{a.s.}{\rightarrow} \psi$, 
and then show how to use this fact to construct an asymptotically 
optimal multigroup CIP. 

The ingredients of our FAB CIP for a single population mean $\theta$ 
are as follows: 
Let $\bar Y\sim N(\theta,\sigma^2/n)$ and
 $q S^2/\sigma^2 \sim  \chi^2_q$ be independent. 
Consider the $1-\alpha$ CRP for $\theta$ given by 
\begin{equation}
  C_w(\bar y, s^2 ) = 
  \{ \theta : \bar y + \tfrac{s}{\sqrt{n}} 
   t_{\alpha(1-w(\theta)) } < \theta < 
   \bar y + \tfrac{s }{\sqrt{n}}  
     t_{1-\alpha w(\theta) } \}, 
\label{eqn:combo}
\end{equation}
where the $t$-quantiles are those of the $t_{q}$-distribution. 
As described in Section 2.2, this procedure has $1-\alpha$ coverage
for every value of $\theta$ and is an interval if 
$w:\mathbb R \rightarrow (0,1)$ is a continuous nondecreasing function.
This holds for non-random
$w$-functions as well as for random  $w$-functions that are 
independent of $\bar Y$ and $S^2$. 
In particular, 
suppose we have estimates $\hat \psi = (\hat \mu , 
  \hat\tau^2,\hat\sigma^2/n)$ that are independent of $\bar Y$ and $S^2$. 
We can then let $w = 
w_{\hat \psi}$,
the $w$-function 
of the Bayes optimal $z$-interval
assuming a prior distribution $\theta\sim N(\hat \mu, \hat\tau^2)$
and that $\Var{\bar Y|\theta} = \hat \sigma^2/n$. 
Note that we are not assuming $(\mu,\tau^2,\sigma^2/n)$ 
actually equals $(\hat \mu,\hat \tau^2,\hat \sigma^2/n)$, we 
are just using these values to 
approximate the optimal $w$-function 
by $w_{\hat \psi}$ and the optimal CIP by 
$C_{w_{\hat \psi}}$.

The random interval $C_{w_{\hat\psi}}(\bar Y,S^2)$ 
differs from the optimal interval $C_{w_\psi}(\bar Y)$ in 
three ways: First, the former uses $S^2$ instead of 
$\sigma^2$ to scale the endpoints of the interval. 
Second, the former uses $t$-quantiles instead of 
standard normal quantiles. Third, the former uses 
$\hat \psi$ to define the  $w$-function, instead 
of $\psi$. 
Now as $q$ increases, $S^2 \stackrel{a.s.}{\rightarrow} \sigma^2$
and the $t$-quantiles in (\ref{eqn:combo}) converge to 
the corresponding $z$-quantiles. 
If we are also in a scenario where $\hat \psi$ can be indexed by $q$ and 
 $\hat \psi \stackrel{a.s.}{\rightarrow} \psi$, then 
we expect that  $w_{\hat \psi}$ converges to 
$w_{\psi}$ and that the risk of $C_{w_{\hat\psi}}$ 
converges to the oracle risk:
\begin{proposition}
Let $\bar Y\sim N(\theta,\sigma^2/n)$, 
    $q S^2/\sigma^2 \sim  \chi^2_q$, and $\hat \psi$ be independent 
    for each value of $q$, 
    with $\hat \psi \stackrel{a.s.}{\rightarrow } \psi$ as $q\rightarrow
\infty$.  Then 
\begin{enumerate}
\item $C_{w_{\hat\psi}}$ defined in (\ref{eqn:combo}) 
is a $1-\alpha$ CIP for each value of $\theta$ and $q$; 
\item $\Exp{ |  C_{w_{\hat\psi}}|} \rightarrow 
         \Exp{ |  C_{w_{\psi}}|}$ as $q\rightarrow \infty$. 
\end{enumerate}
\label{prop:aopttint}
\end{proposition}

We now return to the problem of constructing an asymptotically optimal 
multigroup procedure.
Let $\bar Y_j$ and $S^2_j$ be the sample mean and variance for a given group $j$. 
Divide the remaining groups into two sets, with $p_1-1$ in the first set 
and $p_2=p-p_1+1$ in the second. 
Pool the group-specific sample 
variances of the first 
set of groups with $S_j^2$ to obtain an estimate $\tilde S^2_j$ of 
$\sigma^2$,
so that $ p_1(n-1)\tilde S^2_j/\sigma^2\sim \chi^2_{p_1(n-1)}$. 
From the remaining groups, obtain a strongly consistent estimate $\hat \psi_j$  of $\psi$
(such as the MLE or a moment-based estimate). 
Then $\bar Y_j$, $\tilde S^2_j$ and $\hat \psi_j$ 
are independent for each value of $p$. 
Therefore, a $1-\alpha$ CIP for $\theta_j$ is given by 
\begin{equation}
  C_{w_{\hat\psi_j}}(\bar y_j, \tilde s_j^2 ) = 
  \{ \theta_j : \bar y_j + \tfrac{\tilde s_j}{\sqrt{n}} 
   t_{\alpha(1-w_{\hat\psi_j}(\theta_j)) } < \theta_j < 
   \bar y_j + \tfrac{\tilde s_j}{\sqrt{n}}  
     t_{1-\alpha w_{\hat \psi_j}(\theta_j) } \}, 
\label{eqn:aopttint}
\end{equation}
where the quantiles are those of the $t_{p_1(n-1)}$ distribution. 
If $p_1$ is chosen so that 
it remains a fixed fraction of $p$ as $p$ increases, 
then $\tilde S^2_j$ and $\hat \psi_j$ converge to  
$\sigma^2$ and $\psi$ respectively, and  
the $t$-quantiles converge to the corresponding 
standard normal quantiles. 
By Proposition \ref{prop:aopttint}, the risk of this interval 
converges to that of the oracle risk. 
Repeating this construction for each group $j$ results in 
a multigroup confidence procedure that has $1-\alpha$ coverage 
for each group \emph{conditional} on $(\theta_1,\ldots, \theta_p)$, 
but is also asymptotically optimal on average across the
$N(\mu,\tau^2)$ population of $\theta$-values. 

In practice for finite $p$, different choices of 
$p_1$ and $p_2$ will affect the resulting confidence intervals. 
Since the minimal width of each interval is directly tied 
to the degrees of freedom $p_1(n-1)$ of the
variance estimate  $\tilde S_j^2$, we suggest choosing $p_1$ 
to ensure that
the quantiles of 
the $t_{p_1(n-1)}$ distribution are reasonably close to those 
of the standard normal distribution. 
If either $p$ or $n$ are large, this can be done 
while still allowing $p_2$ to be large enough 
for $(\hat \mu,\hat\tau^2,\hat\sigma^2/n)$ to be useful.

\subsection{A procedure for heteroscedastic groups}
If a researcher is unwilling to assume a common within-group 
variance, constant $1-\alpha$ group-specific coverage
can still be ensured by using  intervals 
of the form 
\begin{equation}
  C_{w_j}(\bar y_j, s^2_j ) = 
  \{ \theta_j : \bar y_j + \tfrac{s_j}{\sqrt{n_j}} 
   t_{\alpha(1-w_j(\theta_j)) } < \theta_j < 
   \bar y_j + \tfrac{s_j }{\sqrt{n_j}}  
     t_{1-\alpha w_j(\theta_j) } \}, 
\label{eqn:hetfabt}
\end{equation}
where $w_j$ is an estimate of the Bayes-optimal $w$-function 
discussed at the end of Section 2.2, 
estimated with data from groups 
other than $j$. We recommend obtaining $w_j$ from a hierarchical 
model for both the group-specific means and variances,
as this allows across-group sharing of information about both 
of these quantities. 
For example, 
the replication material for this article provides code to 
obtain estimates of the 
$w$-function that is optimal for the following
model of across-group heterogeneity:
\begin{align} 
\theta_1,\ldots, \theta_p &\sim \text{i.i.d.} \ N(\mu,\tau^2) 
\label{eqn:hmmv}  \\
1/\sigma_1^2,\ldots, 1/\sigma_p^2 &\sim \text{i.i.d.} \ \text{gamma}(a,b). \nonumber 
\end{align}
We estimate the 
across-group heterogeneity parameters
 $(\mu,\tau^2,a,b)$ as follows:
For each group $j$
 let $X_j^2=\sum_i (Y_{i,j} -\bar Y_j)^2 \sim \sigma^2_j \chi^2_{n_j-1}$. 
If $1/\sigma^2_j\sim \text{gamma}(a,b)$ independently for each $j$
then the marginal density of $X_1^2,\ldots, X_p^2$ can be shown to be 
\[
  p(x_1^2,\ldots, x_p^2|a,b)  
\prod_{j=1}^p  c(x^2_j) 
    \frac{\Gamma(a+ (n_j-1)/2) b^a }{\Gamma(a)
     (b+x^2_j/2)^{a+(n_j-1)/2} }, 
\] 
where $c$ is a function that does not depend on $a$ or $b$. 
This quantity can be maximized 
to obtain marginal maximum likelihood estimates of 
$\hat a$ and $\hat b$.  
Now if $\sigma^2_1,\ldots, \sigma^2_p$ were known, 
then a maximum likelihood estimate of 
$( \mu,\tau^2)$ could be obtained 
based on the fact that 
under the hierarchical model, 
$\bar Y_j \sim N( \mu , \sigma^2_j/n_j + \tau^2)$ independently 
across groups. Since the $\sigma^2_j$'s are not known 
we use empirical Bayes estimates, given by 
  $\hat\sigma^2_j = (\hat b + x_j^2/2)/(\hat a + (n_j-1)/2) $, to 
obtain  the ``plug-in'' marginal likelihood estimates
$(\hat \mu,\hat\tau^2)$:
\[
 (\hat \mu,\hat\tau^2) = \arg \max_{\mu,\tau^2} 
 \prod_j \frac{1}{\sqrt{ \hat \sigma^2_j/n_j + \tau^2}}
 \phi\left ( \frac{\bar y_j-\mu}{\sqrt{\hat \sigma^2_j/n_j + \tau^2}}\right ),
\]
where $\phi$ is the standard normal probability  density function. 

To create a FAB $t$-interval for a given group $j$, we obtain 
estimates $(\hat \mu_j,\hat \tau^2_j,\hat a_j,\hat b_j)$ using 
the procedure  described above with data 
from all groups except group $j$. 
The $w$-function $w_j$ for group $j$ is 
taken to be the Bayes-optimal $w$-function defined by 
Equation \ref{eqn:wtopt}, under the estimated prior 
$\theta_j \sim N(\hat \mu_j ,\hat \tau^2_j)$ and 
$1/\sigma^2_j \sim \text{gamma}(\hat a_j ,\hat b_j)$. 
The independence of $(\bar Y_j,S^2_j)$ and 
$(\hat \mu_j,\hat \tau^2_j ,\hat a_j,\hat b_j)$ 
ensures that the resulting FAB $t$-interval has exact 
$1-\alpha$ coverage, conditional on $\theta_1,\ldots, \theta_p$
and $\sigma_1^2,\ldots, \sigma_p^2$. 

We speculate that this procedure enjoys similar optimality 
properties to those of the approach for homoscedastic groups 
described in Section 3.1: If the hierarchical model 
given by (\ref{eqn:hmmv}) is correct, then as the number $p$ of 
groups  increases, the estimates 
$(\hat \mu_j,\hat \tau^2_j,\hat a_j,\hat b_j)$ will converge to 
$(\mu,\tau^2,a,b)$ and the interval for a given group 
will converge to the corresponding Bayes-optimal interval. 
So far we have been unable to prove this, the primary difficulty 
being that the Bayes-optimal $w$ function
given by Equation $\ref{eqn:wtopt}$ is a non-standard integral 
involving the non-central $t$-distribution, and is not easily 
studied analytically.

\section{Radon data example}

A study by the U.S. Environmental Protection Agency measured radon levels in
a random sample of homes.
\citet{price_nero_gelman_1996} use a subsample of these data to estimate
county-specific mean radon levels (on a log scale) in the state of Minnesota.
This dataset consists of log radon values measured in 919 homes, each
being located in one of $p=85$ counties.
County-specific sample sizes ranged from 1 to 116 homes.
In this section we obtain a $95\%$ FAB confidence interval
for each county-specific mean radon level, based on 
data from all of the counties, and 
compare these intervals to the corresponding UMAU intervals. 
Also, in a simulation study based on these data, we compare 
the expected widths of these two types of intervals to 
empirical Bayes posterior intervals, and show how the latter 
do not provide constant coverage across values of the  county-specific means.

\subsection{County-specific confidence intervals}
Letting $Y_{i,j}$ be the radon measurement for home $i$ in  county $j$, 
we assume
throughout this section that 
$Y_{1,j}\ldots, Y_{n_j,j} \sim $ i.i.d.\ $N(\theta_j, \sigma^2_j)$
and that the data are independently sampled across counties. 
Under the assumptions of a constant across-county variance and the 
normal hierarchical model
$\theta_1,\ldots, \theta_p \sim $ i.i.d.\ $N(\mu,\tau^2)$,
the maximum likelihood estimates of $\sigma^2$, $\mu$ and $\tau^2$ are
$\hat \sigma^2 = 0.637$,
$\hat \mu =1.313$ and $\hat \tau^2 =0.096$. The estimate of the
across-county variability is substantially smaller than the estimate of
within-county variability, suggesting that there is useful information 
to be shared across the groups. 
However, Levene's test of heteroscedasticity (an $F$-test using the  
absolute difference between the data and group-specific medians) 
rejects the null of homoscedasticity with a $p$-value of 0.011. 
For this reason, we use the FAB $t$-interval procedure described in 
Section 3.2 for each group, having the 
form $\{ \theta_j:  \bar y_j +
\sqrt{s^2_j/n_j}\times t_{\alpha(1-w_j(\theta_j))} < \theta_j <  
\bar y_j + \sqrt{s^2_j/n_j}\times t_{1-\alpha w_j(\theta_j)}
\}$, where $\alpha=.05$, 
$\bar y_j$ and $s^2_j$ are the sample mean and variance from 
county $j$, and $w_j$ is the optimal $w$-function 
assuming $\theta_j \sim N( \hat \mu_j, \hat \tau^2_j)$ 
and $1/\sigma^2_j \sim $ gamma$(\hat a_j,\hat b_j)$, 
where $\hat\mu_j,\hat \tau^2_j,\hat a_j$ and $\hat b_j$ are estimated from the 
counties other than county $j$. 
Such intervals have 95\% coverage for each county, 
assuming only within-group normality. 

We constructed FAB and UMAU intervals for each county that had a sample size
greater than one, 
i.e.\ counties for which we could obtain an 
unbiased within-sample variance estimate. 
Intervals for counties with sample sizes greater than two are 
displayed in 
Figure \ref{fig:radon} (intervals based on a sample size of two were 
excluded from the figure because their widths make smaller
intervals difficult to visualize). 
The UMAU intervals are wider than the FAB intervals 
for 77 of the 82 counties having a sample size greater than 1, and are
30\% wider  on average 
across counties. 
Generally speaking, the counties for which the FAB intervals provide 
the biggest improvement are those with smaller sample sizes and 
sample means near the across-group average. Conversely, 
the five counties for which 
the UMAU intervals are  narrower than the FAB  interval are
those with moderate to large sample sizes, and sample means 
somewhat distant from the across-group average.

\begin{figure}
\centering
\includegraphics[width=6.5in]{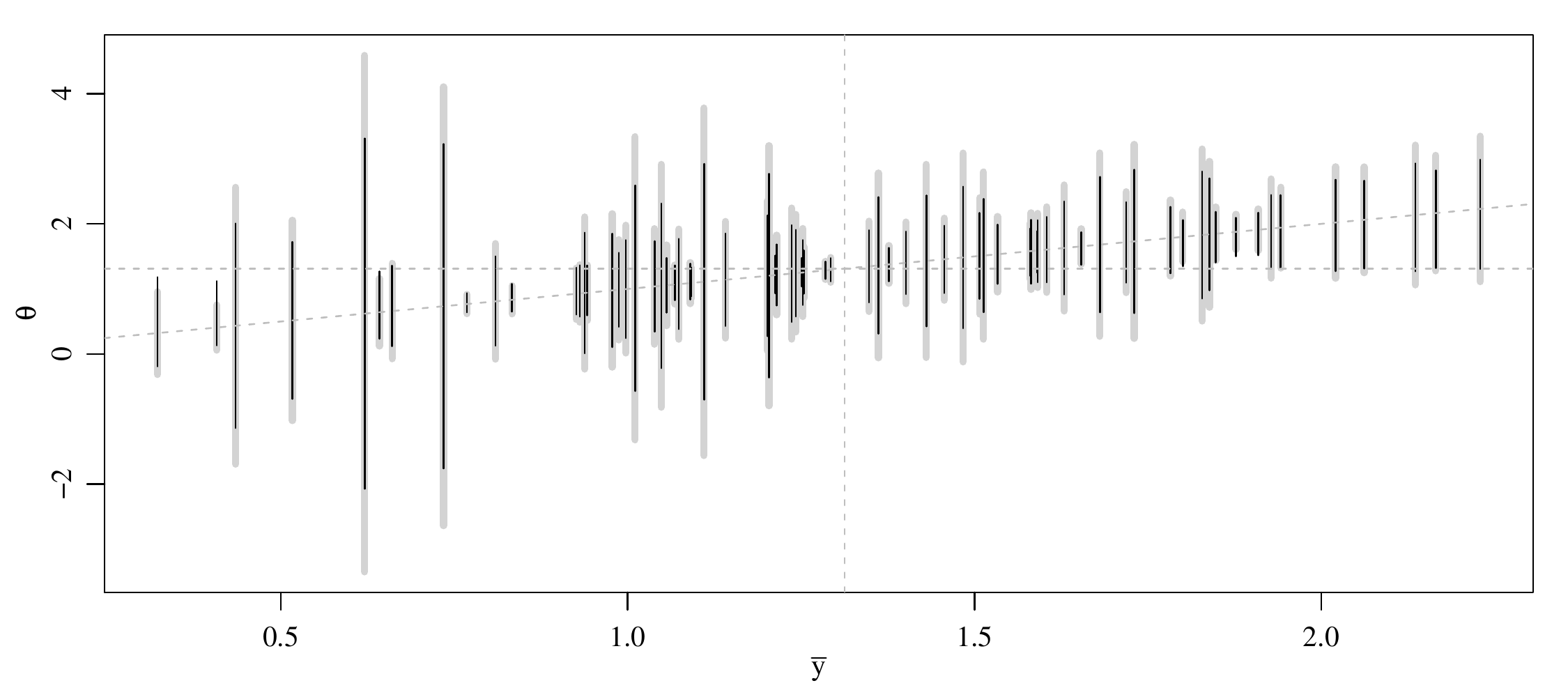}
\caption{FAB and UMAU 95\% confidence intervals for the radon dataset.
   The UMAU intervals are plotted as wide gray lines, the FAB intervals
   as narrow black lines. Vertical and horizontal  lines are drawn at 
   $\sum \bar y_j/p$. 
  }
        \label{fig:radon}
\end{figure}

\subsection{Risk performance and comparison to posterior intervals} 
Assuming within-group normality, 
the FAB interval procedure described above has 95\% coverage 
for each group $j$ and for all values of $\theta_1,\ldots, \theta_p$.  
Furthermore, the procedure is designed 
to approximately 
minimize the expected risk under the hierarchical model  
$\theta_1,\ldots,\theta_p \sim $ i.i.d.\ $N(\mu,\tau^2)$, 
among all $95\%$ CRPs. 
However, one may wonder how well the FAB procedure works 
for fixed values of $\theta_1,\ldots, \theta_p$. 
This question is particularly relevant in cases where 
the hierarchical model is misspecified, or 
if a hierarchical model is not appropriate (e.g., if the 
groups are not sampled). 
We investigate this for the radon data with a simulation study in 
which we take the county-specific sample means and variances 
as the true county-specific values, that is, 
we set $\theta_j = \bar y_j$ and $\sigma^2_j = s_j^2$ for each 
county $j$. We then simulate $n_j$ observations 
for each county $j$ from the model 
$Y_{1,j},\ldots, 
 Y_{n_j,j}\sim $ i.i.d.\ $N(\theta_j,\sigma^2_j)$. 

We generated 10,000 such simulated datasets. 
For each dataset, we computed the widths of the 
95\%  FAB and UMAU confidence 
intervals for each county having a sample size greater than one. 
Additionally, for comparison we also computed 
empirical Bayes posterior 
intervals, which are often used in hierarchical modeling. 
The posterior interval for 
group $j$ is given by 
$\hat \theta_j  \pm t_{1-\alpha/2 }  \times 
     (1/\hat \tau^2 +n_j/\hat s^2_j )^{-1/2}$, where 
$\hat \theta_j$ is the empirical Bayes estimator given by 
\[
\hat\theta_j  = \frac{  \hat \mu/ \hat \tau^2 + \bar y_j n_j/ s_j^2}
                      { 1/\hat\tau^2 + n_j/ s_j^2 }, 
\]
and 
$t_{1-\alpha/2}$ is the $1-\alpha/2$ quantile of the $t_{n-1}$-distribution.
As discussed in the Introduction, such intervals are always narrower than
the corresponding UMAU intervals but will not 
have $1-\alpha$ frequentist coverage for each group. Instead, 
such intervals generally have $1-\alpha$ coverage on average, 
or in expectation with respect to the hierarchical model over 
the $\theta_j$'s. 

The results of this simulation study are displayed in Figure 
\ref{fig:radon_sim}. The left panel of the figure 
gives the expected widths of the FAB and Bayes 
procedures relative to those of the UMAU procedure.  
Based on the 10,000 simulated datasets, 
the estimated expected widths across counties were
about 
2.28, 1.60 and 1.61, 
respectively for the UMAU, FAB and Bayes procedures respectively. 
As with the 
actual interval widths for the non-simulated data, 
expected widths of the FAB intervals are smaller than 
those of the UMAU intervals for most counties (79 out of 82). 
The Bayes intervals are always narrower than the 
UMAU intervals for all groups by construction. However, 
while they tend to be narrower than the FAB intervals
for $\theta_j$'s far from $\bar\theta=\sum \theta_j/p$, near this average 
they are often wider than the FAB intervals. This is not 
too surprising - the FAB intervals are at their narrowest near this 
overall average, while 
the Bayes intervals tend to over-cover here. 
This latter issue is illustrated in the right panel of the 
figure,
which shows how the Bayes credible intervals do not have constant 
coverage across groups. 
This is because
the Bayes intervals are centered around biased 
estimates that are shrunk towards the estimated overall mean $\bar \theta$. 
If $\theta_j$ is far from $\bar\theta$ then the bias is high and 
the coverage is too low, whereas if $\theta_j$ is near $\bar\theta$ 
the coverage is too high since the variability of 
the shrinkage estimate
$\hat\theta_j$ is lower than that of $\bar y_j$.  
The group-specific coverage rates of the Bayes intervals 
vary from about 91\% to 98\%, 
although the average coverage rate 
across groups is approximately 95\%. 
In summary, 
the UMAU procedure provides constant $1-\alpha$ coverage across 
groups, but wider intervals  than those obtained from the FAB and Bayes 
procedures. 
The Bayes procedure provides narrower intervals but non-constant coverage. 
The FAB 
procedure provides both narrower intervals and 
constant coverage.

\begin{figure}
\centering
\includegraphics[width=6.5in]{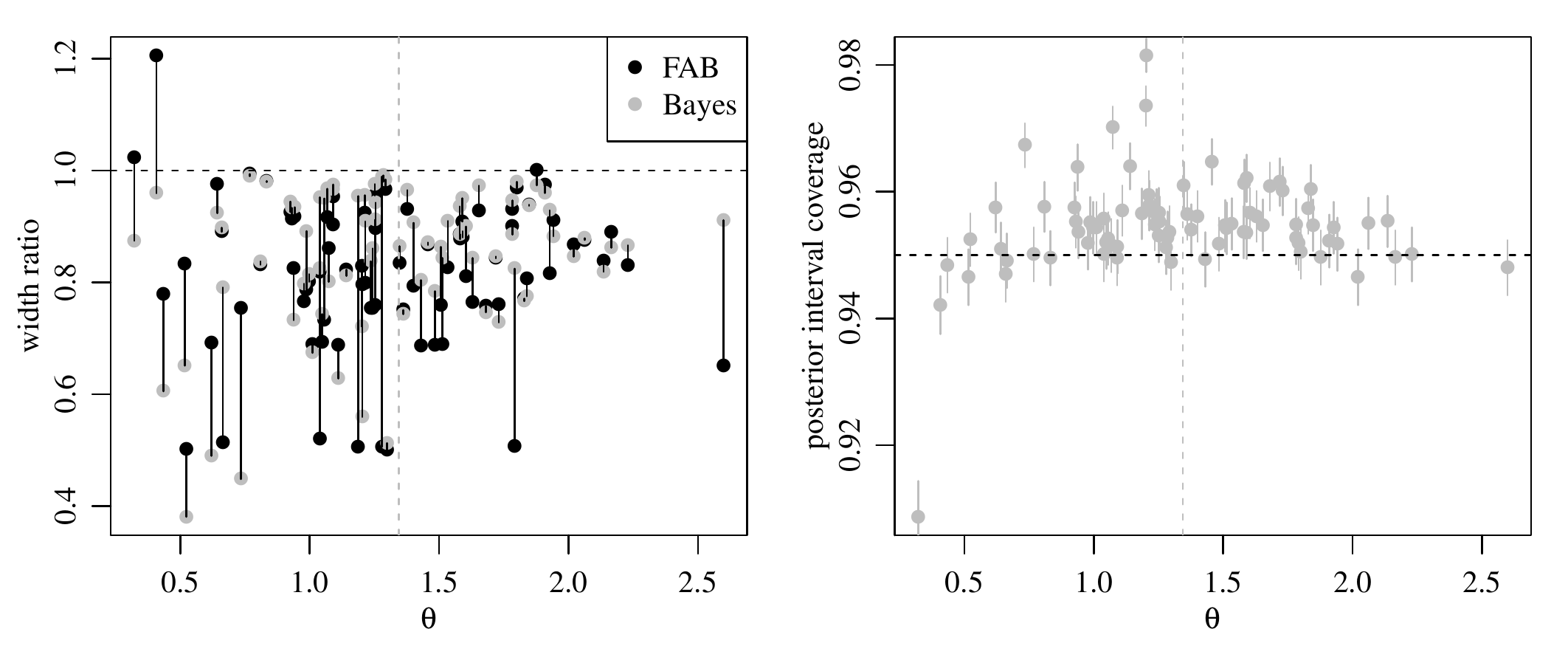}
\caption{Simulation results. The left panel 
gives relative expected interval widths of the FAB and 
Bayes procedures relative to the UMAU procedure. 
The right panel indicates how coverage rates 
of Bayes posterior intervals
are not constant across groups. 
Points are coverage rates based on 10,000 simulated 
datasets, and vertical lines are nominal 95\% intervals 
representing Monte Carlo standard error. 
Vertical lines are drawn at  $\sum \theta_j/p$ in each panel.
  }
 \label{fig:radon_sim}
\end{figure}

\section{Discussion}
Standard analyses of multilevel data utilize 
multigroup confidence interval procedures
that either have constant coverage but do not 
share information across groups, or share information 
across groups but lack constant coverage. 
These latter procedures typically do maintain 
a pre-specified coverage rate on average across groups, 
but the value of this property is unclear 
if one wants to make group-specific inferences. 
The FAB procedures 
developed in this article 
have coverage rates that are constant in the mean parameter, 
and so maintain constant coverage for each group 
selected
into the dataset, while
also making use of across-group information. 
The FAB procedures are  approximately optimal among constant coverage procedures
if the across-group heterogeneity is well-represented 
by a normal hierarchical model. 

If the across-group heterogeneity is not well-represented by 
a hierarchical normal model, then the FAB procedure 
will still maintain the chosen constant coverage rate but may not 
be optimal. 
We speculate that in such cases, 
the FAB procedure based on a hierarchical normal model, while not optimal,
will still have better risk than the UMAU procedure when the 
across-group heterogeneity corresponds to any probability distribution with a 
finite second moment. This is partly because the UMAU procedure is 
a limiting case of the FAB procedure as the across-group 
variance goes to infinity. 
We have developed an analytical argument of this and have 
gathered computational evidence, 
but a complete proof 
of the dominance of a misspecified FAB procedure over the UMAU 
procedure is still a work in progress. 
Of course, the basic idea behind the FAB procedure could be 
implemented with alternative models describing across-group heterogeneity, 
such as 
models that allow for sparsity 
among the group-level parameters. 
We have implemented a few such procedures computationally, but 
studying them analytically is challenging. 

\medskip

Replication code for this article can be found at the second author's 
website. The  multigroup FAB procedures discussed 
in Sections 3 and 4 are provided by the {\sf R}-package
{\tt fabCI}.


\appendix

\section*{Appendix: Proofs}

\begin{proof}[Proof of Lemma \ref{lemma:twoside}]
	This lemma follows from \citet[Section 5.3]{Ferguson_1967}, 
which says that 
for any level-$\alpha$  test of a point null hypothesis
for a one-parameter exponential family, there
exists a two-sided test of equal or greater power. 
Let $\{\tilde \phi_\theta(y):\theta\in \mathbb R\}, \{\tilde A(\theta): \theta\in \mathbb R\}$ be the test functions and 
acceptance regions corresponding to the CRP $\tilde C(y)$.  
The coverage of $\tilde C$ is 
	\begin{equation}
		\Pr(Y \in \tilde A(\theta) | \theta) = 1- \Exp{\tilde \phi_{\theta}(Y)|\theta}.
		\label{eq:rateC1}
	\end{equation}
	By Theorem 2 from \citet[Section 5.3]{Ferguson_1967}, 
for each $\theta\in \mathbb R$ there exists a two-sided test $\phi_{\theta}$ such that
	\begin{equation}
		\Exp{\phi_{\theta}(Y) | \theta} = \Exp{\tilde \phi_{\theta}(Y)| \theta}.
		\label{eq:th21}
	\end{equation}
Denote the acceptance regions corresponding to these two-sided test as $\{ A(\theta):\theta\in \mathbb R\}$. 
Inverting these regions gives a CIP $C(y)$. The coverage of  $C(y)$ is 
\begin{equation}
\Pr(Y \in A(\theta) | \theta) = 1- \Exp{\phi_{\theta}(Y)|\theta}.  
	\label{eq:rateC2}
\end{equation}
	Hence by (\ref{eq:rateC1}), (\ref{eq:rateC2}), and (\ref{eq:th21}), the coverage of $C(y)$ is the same as the coverage of $\tilde C(y)$.
The width of $\tilde C(y)$ is:
\begin{equation*}
W(y) = \int_{\mathbb{R}} 1(t \in \tilde C(y)) dt = \int_{\mathbb{R}} 1(y \in \tilde A(t)) dt = \int_{\mathbb{R}} (1 - \tilde \phi_t(y)) dt. 
	\end{equation*} 
The expected width of $\tilde C(y)$ is:
	\begin{equation}
		\Exp{\tilde{W}|\theta } = \int_{\mathbb{R}} W(y) p(y | \theta) dy = \int_{\mathbb{R}} \int_{\mathbb{R}}  (1 - \tilde \phi_t(y))  p(y | \theta) dy dt 
		\label{eq:ew1}
	\end{equation} 
where $p(y|\theta)$ is the density of $Y$ given $\theta$. 
Similarly, the expected width of $C(y)$ is
	\begin{equation}
\Exp{W|\theta} =  \int_{\mathbb{R}} \int_{\mathbb{R}}  (1 - \phi_t(y))  p(y | \theta) dy dt. 
		\label{eq:ew2}
	\end{equation} 
Again, by Theorem 2 from \citet[Section 5.3]{Ferguson_1967}, for every $\theta \in \mathbb{R}$ 
	\begin{equation*}
		\int_{\mathbb{R}}   \phi_t(y) p(y | \theta) dy \geq \int_{\mathbb{R}}   \tilde \phi(y) p(y | \theta) dy.
	\end{equation*} 
	Thus 
	\begin{equation}
		\int_{\mathbb{R}}  (1- \phi_t(y)) p(y | \theta) dy \leq \int_{\mathbb{R}}  (1- \tilde \phi_t(y)) p(y | \theta) dy.
		\label{eq:ew3}
	\end{equation} 
	Therefore by (\ref{eq:ew1}), (\ref{eq:ew2}), (\ref{eq:ew3}), we have $\Exp{W |\theta}  \leq \Exp{\tilde W |\theta}$.
\end{proof}

\begin{proof}[Proof of Proposition \ref{prop:C}]
	Without loss of generality, we prove the proposition for the simple case when $\mu = 0$ and $\sigma^2 = 1$. 
Other cases can be obtained by reparametrizing as  
$\tilde Y=(Y-\mu)/\sigma$, $\tilde \theta= (\theta-\mu)/\sigma$ and 
$\tilde \tau^2 = \tau^2/\sigma^2$ so that 
$\tilde Y\sim N(\tilde\theta,1)$ and $\tilde \theta\sim N(0,\tilde \tau^2)$. 

The Bayes optimal procedure minimizes the Bayes risk 
		$R(\psi ,C_w) = \int \Pr( Y\in A(\theta) ) \, d\theta$, 
where $Y$ has the marginal density $N(0,1+\tau ^2)$. 
For a given $w$-function, the Bayes risk is
	\begin{equation}
		\begin{aligned}
			R(\psi ,C_w) & = \int_{\mathbb{R}} \Phi(\frac{\theta - l}{\sqrt{1+\tau ^2}}) - \Phi(\frac{\theta - u}{\sqrt{1+\tau ^2}}) \ d\theta \\
			& = \int_{\mathbb{R}} \Phi(\frac{\theta- \Phi^{-1}(\alpha(1-w))}{\sqrt{1+\tau ^2}}) - \Phi(\frac{\theta - \Phi^{-1}(1-\alpha w)}{\sqrt{1+\tau ^2}}) \ d\theta. 
		\end{aligned}
		\label{eq:integrand}
	\end{equation}
We will show that, as a function of $w$,  the integrand $H$ 
is minimized at $w_{\psi}(\theta)$
as given in the  proposition statement. 
First, we obtain the derivative of $H$ with respect to $w$: 
	\begin{equation*}
		\begin{aligned}
			H'(w) =  	 
			& \exp(-\frac{1}{2}\frac{(\theta - \Phi^{-1}(\alpha(1-w)))^2}{1+\tau ^2})\frac{1}{\sqrt{1+\tau ^2}}\frac{\alpha}{\exp(-\frac{1}{2}(\Phi^{-1}(\alpha(1-w)))^2)}\\
			& - \exp(-\frac{1}{2} \frac{(\theta - \Phi^{-1}(1-\alpha w))^2}{1+\tau ^2})\frac{1}{\sqrt{1+\tau ^2}}\frac{\alpha}{\exp(-\frac{1}{2}(\Phi^{-1}(1-\alpha w))^2)}.
		\end{aligned}
	\end{equation*}
	Setting this to zero and simplifying shows that a critical point $w_\psi$ satisfies
	\begin{equation}
		{2\theta }/{\tau ^2} = \Phi^{-1}(w\alpha) - \Phi^{-1}((1-w)\alpha). 
		\label{eq:derivative}
	\end{equation}
Let the right side of (\ref{eq:derivative}) be $g(w)$.
	It's not difficult to verify that $g(w)$ a continuous and strictly increasing function of $w$, with range $(-\infty, \infty)$. Thus there is a unique solution $ w_{\psi}(\theta)$ to the equation above, $w_{\psi}(\theta) = g^{-1}({2\theta / \tau ^2})$, which is a continuous and strictly increasing function of $\theta$. Since $H'(w)$ is continuous on $(0,1)$ with only one root, and $\lim_{w \rightarrow 0} \hspace{0.2cm}  H'(w)= -\infty$, $\lim_{w \rightarrow 1} \hspace{0.2cm}  H'(w) = \infty$,
	then $H(w)$ is minimized by $w_{\psi}(\theta)$. Therefore $w_{\psi}(\theta)$ minimizes the Bayes risk, and $C_{w_{\psi}}$ is the Bayes-optimal procedure among all CRPs. 
\end{proof}

\begin{proof}[Proof of Lemma \ref{lemma:incint}]
$C_w(y)$ can be written as $C_w(y) = \{ \theta : y< \theta - \sigma l(\theta) \; \textrm{and} \; \theta  - \sigma u(\theta) < y \}$. Letting $f_1(\theta) =  \theta  - \sigma u(\theta)$, $f_2(\theta) =  \theta - \sigma l(\theta)$, we first prove that $C_w(y)$ can also be written as $C_w(y) = \{ \theta : f_2^{-1}(y)<\theta \; \textrm{and} \; \theta  < f_1^{-1}(y) \}$.
	Note that both $\Phi^{-1}$ and $w(\theta)$ are continuous nondecreasing functions. Therefore $f_1(\theta) = \theta - \Phi^{-1} ( 1-\alpha w(\theta))$ is a strictly increasing continuous function, with 
$\lim_{\theta \to -\infty} \theta - \Phi^{-1} ( 1-\alpha w(\theta)) = -\infty$
and 
$\lim_{\theta \to +\infty} \theta - \Phi^{-1} ( 1-\alpha w(\theta))= +\infty.$
	Hence, $f_1^{-1}$ exists, and is a strictly increasing continuous function with range $(-\infty, \infty)$. Thus $f_1(\theta) < y$ can also be expressed as $\theta  < f_1^{-1}(y)$. Similarly, $y < f_2(\theta) $ can also be expressed as $ f_2^{-1}(y) < \theta$. Next, in order to show that $C_w(y)$  is an interval, we need to show that $f_2^{-1}(y) < f_1^{-1}(y)$.  To see this, we only need to show 
	\begin{equation*}
		\theta - \sigma \Phi^{-1}(1-\alpha w(\theta)) < \theta - \sigma \Phi^{-1}(\alpha(1-w(\theta))),
		\label{eq:UL}
	\end{equation*}
	or that $\Phi^{-1}(\alpha w(\theta)) <  \Phi^{-1}(1-\alpha(1-w(\theta)))$. This follows since $\Phi^{-1}(x)$ is a strictly increasing function. Thus $C_w(y) = \{ \theta : f_2^{-1}(y)<\theta  < f_1^{-1}(y) \}$, which is an interval.
\end{proof}

\begin{proof}[Proof of Lemma \ref{lemma:tincint}]
	The proof is basically the same as the proof of Lemma 2.2. We only need to replace $y$ with $\bar y$, $\sigma$ with $s/\sqrt{n}$, and the $z$-quantiles with  $t$-quantiles, and then use the same logic as in the proof of Lemma 2.2. 
\end{proof}

The proof of Proposition \ref{prop:aopttint} requires the following 
lemma that bounds the width
of the FAB $t$-interval:
\begin{lemma*} 
	The width of $C_{w_{\psi }}(\bar y, s^2)$ satisfies
	\begin{equation}
		|C_{w_{\psi }}(\bar y, s^2) |  < |\bar y-\mu| + \tfrac{s}{\sqrt{n}}(  |t(\alpha/2)|+|t(1-\alpha/2)|), 
		\label{bound}
	\end{equation}
	where $t$-quantiles are those of the $t_q$-distribution. 
\end{lemma*}

\begin{proof}
	For notational convenience,
	for this proof and the proof of Proposition 3.1, we write $t_\alpha$ as $t(\alpha)$. By previous results, the endpoints $\theta^L$ and $\theta^U$ of $C_{w_{\psi}}(\bar y, s^2)$ are solutions to
	\begin{equation}\label{eq:endpoints}
		\begin{array}{l}
			\theta^U - \tfrac{s}{\sqrt{n}}t(1- \alpha w_{\psi}(\theta^U) )=\bar y  \\
			\theta^L - \tfrac{s}{\sqrt{n}}t(\alpha (1-w_{\psi }(\theta^L))) =\bar y.\\
		\end{array}
	\end{equation} 
	Here $w_{\psi }(\theta)$ is defined as $
			w_{\psi }(\theta)  = g^{-1} ( \tfrac{2(\theta-\mu)}{ {\tau}^2/\sigma} )$, where $g(w)  =  \Phi^{-1}(\alpha w) - \Phi^{-1}(\alpha(1-w) )$.
	At the upper endpoint, we have $w_{\psi}(\theta^U) = F((\bar y-\theta^U)/(s/\sqrt{n}))   /\alpha$,
	where $F$ is the CDF of the $t_q$-distribution.
When $\theta^U > \mu$, we have $w_{\psi}(\theta^U) > g^{-1}(0) = 1/2$. Thus $\theta^U < \bar y - \tfrac{s}{\sqrt{n}}t(\alpha/2)$. Also, $g^{-1} ( \tfrac{2(\theta^U-\mu)}{\tau^{\prime2}/\sigma}) < 1$, so $\theta^U > \bar y - \tfrac{s}{\sqrt{n}}t(\alpha)$. When $\theta^U < \mu$, $\bar y - \tfrac{s}{\sqrt{n}}t(\alpha/2)< \theta^U$. This implies that
	\begin{align*}
		&\bar y - \tfrac{s}{\sqrt{n}}t( \alpha)< \theta^U < \bar y  - \tfrac{s}{\sqrt{n}}t( \alpha /2 )          & \ \text{if} \ \theta^U > \mu\\
		&\bar y - \tfrac{s}{\sqrt{n}}t( \alpha /2)<  \theta^U <\mu          & \ \text{if} \  \theta^U < \mu. 
	\end{align*}
	Similarly we have	
	\begin{align*}
		& \mu <\theta^L < \bar y - \tfrac{s}{\sqrt{n}}t(1 - \alpha /2 )          & \ \text{if} \   \theta^L  > \mu\\
		& \bar y - \tfrac{s}{\sqrt{n}}t(1 - \alpha /2)<\theta^L < y - \tfrac{s}{\sqrt{n}}t(1 - \alpha )          & \ \text{if} \  \theta^L < \mu. 
	\end{align*}
	Therefore
	\begin{equation*}
		|C_{w_\psi}(\bar y,s ) | =\theta^U - \theta^L < |\bar y-\mu| + \tfrac{s}{\sqrt{n}}(  |t(\alpha/2)|+|t(1-\alpha/2)|).
	\end{equation*}
\end{proof}

\begin{proof}[Proof of Proposition \ref{prop:aopttint}] 
That $C_{w_{\hat \psi}}$ is a $1-\alpha$ CIP follows by construction 
of the interval and that $\hat \psi$ is independent of $\bar Y$ and $S^2$. 
To prove the convergence of the risk, we denote the endpoints of the oracle CIP $C_{w_\psi}$ as $\theta^U$ and $\theta^L$, which are the solutions to
	\begin{equation*} 
		\begin{array}{l}
			\theta^U - \tfrac{\sigma}{\sqrt{n}}\Phi^{-1}(1- \alpha w_{\psi}(\theta^U) )=\bar Y  \\
			\theta^L - \tfrac{\sigma}{\sqrt{n}}\Phi^{-1}( \alpha (1-w_{\psi }(\theta^L))) =\bar Y.\\
		\end{array}
	\end{equation*} 
We denote the endpoints of $C_{w_{\hat \psi}}$ as $\theta^U_q$ and $\theta^L_q$, which  are the solutions to
	\begin{equation*} 
		\begin{array}{l}
			\theta^U_q - \tfrac{S}{\sqrt{n}}t( 1-\alpha w_{\hat \psi}(\theta^U_q) )=\bar Y  \\
			\theta^L_q - \tfrac{S}{\sqrt{n}}t( \alpha (1-w_{\hat \psi }(\theta^L_q))) =\bar Y.\\
		\end{array}
	\end{equation*} 
We first prove that $|C_{w_{\hat \psi}}|-|C_{w_{ \psi}}| = (\theta^U_q - \theta^U) + (\theta^L - \theta^L_q) \stackrel{a.s.}{\rightarrow } 0$ as $q \to \infty$ for each fixed $\bar Y$.
	We can write the upper endpoints as $\theta^U = G(\psi, \bar Y, \sigma^2)$, and $\theta^U_q = G_q(\hat \psi, \bar Y, S^2)$, where $G$ and $G_q$ are continuous functions of their parameters. The difference between $G$ and $G_q$ is that the former is obtained based on $z$-quantiles, while the later is based on $t$-quantiles.   
	We have
	\begin{align}
		|\theta^U_q - \theta^U| &=|G_q(\hat \psi, \bar Y, S^2) - G(\psi, \bar Y, \sigma^2)| \\ 
		& \leq |G_q(\hat \psi, \bar Y, S^2) - G(\hat \psi, \bar Y, S^2)| + |G(\hat \psi, \bar Y, S^2) - G(\psi, \bar Y, \sigma^2)|. 	
		\label{twocomp}
	\end{align}  
The first term in 
 (\ref{twocomp}) 
converges to zero because the convergence of $G_q \to G$ is uniform, 
and the second term converges to zero because 
$(\hat \psi, S^2) \stackrel{a.s.}{\rightarrow } (\psi,\sigma^2)$. 
Elaborating on the convergence of the first term, 
note that $G_q$ is a monotone sequence of continuous functions: 
Given $q_2 > q_1$, we have $t_{q_2}(1-\alpha w) < t_{q_1}(1- \alpha w)$. Hence $\theta - \tfrac{S}{\sqrt{n}}t_{q_2}(1-\alpha w_{\hat \psi}(\theta)) > \theta - \tfrac{S}{\sqrt{n}}t_{q_1}(1-\alpha w_{\hat \psi}(\theta))$. Therefore $G_{q_2}(\hat \psi, \bar Y, S^2) < G_{q_1}(\hat \psi, \bar Y, S^2)$, and so
by Dini's theorem, $G_q \rightarrow  G$ uniformly on 
a compact set of $(\hat \psi,S^2)$ values.  
Since
 $(\hat \psi, S^2) \stackrel{a.s.}{\rightarrow } (\psi,\sigma^2)$, 
with probability one there is an integer $Q$ such that 
when $q> Q$, $|\hat \psi -\psi| \leq c_1$ and $|S^2 - \sigma^2| \leq c_2$
for any to positive constants $c_1$ and $c_2$. 
Thus, $G_q$ converges to $G$ uniformly on this  compact set and the first 
term in (\ref{twocomp}) converges to zero.

Now we show the expected width converges to the oracle width by integrating over $\bar Y$. 
This is done by finding a 
dominating function for $| C_{w_{\hat\psi}}(\bar Y, S^2)|$
and applying the dominated convergence theorem.  
By the previous lemma we know that 
	\begin{equation*}
		|C_{w_{\hat \psi }}(\bar Y, S^2) |  < |\bar Y|+ |\hat \mu| + \tfrac{S}{\sqrt{n}}(  |t(\alpha/2)|+|t(1-\alpha/2)|).
		\label{bound2}
	\end{equation*}
	Note that $|t(\alpha/2)|+|t(1-\alpha/2)| < |t_1(\alpha/2)|+|t_1(1-\alpha/2)|$, where $t_1$ is the $t$-quantile with one degree of freedom. Similar to the argument earlier in this proof, given two constants $c_1, c_2 > 0$, we can find a $Q$ such that when $q > Q$, we have $|\hat \mu |<| \mu| + c_1$ and $S^2 < \sigma^2 + c_2$ a.s.. Now we have an dominating function for $|C_{w_{\hat \psi }}(\bar Y, S^2) |$
	\begin{equation*}
 |C_{w_{\hat \psi }}(\bar Y, S^2) | < \bar W(\bar Y, S^2,\hat \psi ) = |\bar Y| + |\mu| + c_1 + \tfrac{\sqrt{\sigma^2+c_2}}{\sqrt{n}}(  |t_1(\alpha/2)|+|t_1(1-\alpha/2)|).
	\end{equation*}
	Since $|\bar Y|$ is a folded normal random variable with finite mean, it's easy to see that this dominating function is integrable. Therefore, by dominated convergence theorem we have $
		\lim_{q\rightarrow \infty} \Exp{ |  C_{w_{\hat\psi}}|} =
		\Exp{ |  C_{w_{\psi}}|}$.

\end{proof}


\bibliography{paper}

\end{document}